\documentclass[smallextended]{svjour3}  

\usepackage[parfill]{parskip}
\usepackage{graphicx}
\usepackage{epstopdf}
\usepackage{xargs} 
\usepackage{ifthen}
\usepackage{extpfeil}
\usepackage{tikz}
\usepackage{dsfont}
\usepackage{pgf}
\usepackage{stmaryrd}
\usepackage[verbose]{wrapfig}
\usepackage{xcolor}
\usepackage{url}
\usepackage{tabularx, arydshln,tablefootnote,threeparttable,hhline, diagbox}
\usetikzlibrary{arrows}
\usepackage{amssymb}
\usepackage{hyperref}
\usepackage{amsmath}
\usepackage{geometry}
\usepackage{arydshln}
\usepackage{algorithm}
\usepackage{amsmath, enumerate,array,amsfonts, mathtools,verbatim,lipsum,color}

\newtheorem{defn}{Definition}

\newtheorem{prop}{Proposition}
\renewcommand{\l}{\left}
\renewcommand{\r}{\right}
\newcommand{\eps}{\epsilon}
\newcommand{\ba}{\begin{array}}
\newcommand{\ea}{\end{array}}
\newcommand{\beq}{\begin{equation}}
\newcommand{\eeq}{\end{equation}}
\newcommand{\beqa}{\begin{eqnarray}}
\newcommand{\eeqa}{\end{eqnarray}}
\newcommand{\sg}{\text{sgn}}
\newcommand{\Diag}{\Delta}
\newcommand{\diag}{\delta}
\newcommand{\round}{\text{round}}
\DeclareMathOperator*{\argmin}{arg\,min}
\DeclareMathOperator*{\argmax}{arg\,max}
\begin{document}
%%%%%%%%%%%%%%%%%%%%%%%%%%%%
% This title should be short, because the font is huge, but if you really really
% want to make it bigger you can put \Large or \normalsize before the title.
\title{Practical Integer-to-Binary Mapping for Quantum Annealers}
% This one can be a bit longer if you want. You can also remove it if you don't
% want it.
%\tinytitle{A smaller title in blue.}
\subtitle{}
\author{Sahar Karimi \and Pooya Ronagh}

\institute{ Sahar Karimi (corresponding author)  \at 
1QB Information Technologies (1QBit)\\
\email{sahar.karimi@gmail.com}
\and 
Pooya Ronagh \at 
1QB Information Technologies (1QBit)\\
\email{pooya.ronagh@1qbit.com}
 }

\date{Received: date / Accepted: date}
\maketitle
%%%%%%%%%%
\begin{abstract}
Recent advancements in quantum annealing hardware and numerous studies in this area suggests that quantum annealers have the potential to be effective in solving unconstrained binary quadratic programming problems. Naturally, one may desire to expand the application domain of these machines to problems with general discrete variables. In this paper, we explore the possibility of employing quantum annealers to solve unconstrained quadratic programming problems over a bounded integer domain. We present an approach for encoding integer variables into binary ones, thereby representing unconstrained integer quadratic programming problems as unconstrained binary quadratic programming problems. To respect some of the limitations of the currently developed quantum annealers, we propose an integer encoding, named \emph {bounded-coefficient encoding}, in which we limit the size of the coefficients that appear in the encoding. Furthermore, we propose an algorithm for finding the upper bound on the coefficients of the encoding using the precision of the machine and the coefficients of the original integer problem. Finally, we experimentally show that this approach is far more resilient to the noise of the quantum annealers compared to traditional approaches for the encoding of integers in base two. 

\end{abstract}

\section{Introduction}\label{sec:intro}
Adiabatic quantum computation (AQC) has been proposed as a successful technique for solving certain classes of optimization problems (see \cite{albash2016,farhi01,McGeoch2013,optimization_AQO}).  In particular, quantum annealers (QA) such as the ones manufactured by D-Wave Systems Inc. \cite{DwaveNature} approximate the adiabatic evolution in the presence of various sources of noise and finite temperature to solve Ising models of the form
\begin{equation} \label{Ising_def} \tag{Ising}
	\min_{s \in \mathbb \{\pm 1\}^n}\, s^t J S + h^t s\,,
\end{equation}
where diagonal entries of $J$ are zeroes, and the nonzero entries of $J$ create a subgraph of a sparse graph called a {\em Chimera} graph \cite{chimera}. Many researchers have explored the applicability of quantum annealing to more-general optimization problems. Solving Ising models without the Chimera graph structure on a QA is discussed in \cite{AidanRoy2014}. Degree reduction techniques, such as the one described in \cite{Ishikawa2011}, enable us to solve unconstrained binary polynomial programming problems of higher order using a QA. Several studies \cite{Venturelli2015,Gili,Dominic2015,Zick} have employed quantum annealing for solving real-world applications, and \cite{karimi_GD} discusses the possibility of using QAs for solving constrained problems. 

Many optimization problems, on the other hand, involve integer-valued variables beyond binary values. Examples of such problems could be the number of vehicles or products travelling along each route of a supply chain (see \cite{supplyChain}), or the number of traded assets in portfolio optimization (see \cite{portfolioOpt}). This paper focuses on solving quadratic problems on bounded integer domains, alternatively referred to as unconstrained integer quadratic programming (UIQP) problems, i.e., problems of the following form:

\begin{equation}
\label{UIQP} 
\tag{UIQP}
\begin{array}{lll}
\min &  x^T Q x + q^t x,\\
\text{s.t.} &x_i \in \{0, 1, 2, \ldots, \kappa^{x_i} \} & \quad   \text{for }  i=\{1,2, \ldots, n\},  \nonumber \\ 
\end{array}
\end{equation}

where $\kappa^{x_i} \in \mathbb Z_+$ is an upper bound on $x_i$; here, $\mathbb Z_+$ denotes the set of non-negative integers. Note that in problems where $x_i \in \{ \alpha, \alpha+1, \alpha+2, \ldots, \alpha+\kappa^{x_i}\}$, we may shift $x_i$ and substitute it with $x_i - \alpha$; hence, without loss of generality, we may assume that the domain of $x_i$ starts at 0. 
It is worth mentioning that, although current QAs are limited to representing Ising models, an alternative approach to encoding bounded integer variables in terms of spin variables is to consider physical implementations of quantum processors as representing Potts-Ising models \cite{AQC_qudits,katzgraber_qudits}. 
Our approach is different and is based on reformulating a UIQP problem as an Ising model. To this end, we represent each integer variable as a linear combination of several binary variables, i.e., 
\begin{equation}
x_i = \sum_{j=1}^{d^{x_i}} c^{x_i}_j y^{x_i}_j = (c^{x_i})^t y^{x_i}, 
\label{intEncoding}
\end{equation}
where $c_j^{x_i} \in \mathbb Z_+$, $y_j^{x_i} \in \{ 0, 1\}$ for $j \in \{1, \ldots, d^{x_i}\}$, and the superscript $x_i$ is used for clarity to denote that $c^{x_i}$ and $y^{x_i}$ correspond to the encoding of variable $x_i$. This representation is referred to as {\emph {integer encoding}}. Some of the well-known integer encodings are binary and unary encodings in which $c_j^{x_i} = 2^{j-1}$ and $c_j^{x_i}=1$, respectively. The \emph{width} of an integer encoding, denoted by $d^{x_i}$ in \eqref{intEncoding}, refers to the number of binary variables required for encoding integer variable $x_i$. The width of binary and unary encodings are $\left \lfloor \log_2(\kappa^{x_i}) \right \rfloor+1$ and $\kappa^{x_i}$, respectively. Since a binary variable $y_j$ could be represented with a spin variable $s_j$ via the affine transformation
\begin{equation}\label{binToSpin}
y_j = \frac{1}{2} \l ( s_j  + 1\r),  
\end{equation}
integer variable $x_i$ could be encoded into several spin variables as
\begin{equation}
\label{int2SpinEncoding}
x_i = \frac{1}{2} \l ( \sum_{j=1}^{d^{x_i}} c_j^{x_i} + \sum_{j=1}^{d^{x_i}} c_j^{x_i} s_j^{x_i} \r), 
\end{equation}
which enables us to successfully represent a UIQP as an Ising model. 

After derivation of an Ising model equivalent to \eqref{UIQP}, heuristic methods including quantum annealing could be employed to find the ground state of the problem. The performance of the heuristic methods highly depends on the energy landscape of the problem. In particular, landscapes with tall barriers are challenging for the hill-climbing heuristics, and having wide barriers could impact the performance of methods like quantum annealing that could potentially benefit from quantum tunnelling (see \cite{katzgraber2015} and references therein for more details). 
Having coefficients that are different in orders of magnitude is undesirable, as they could create energy landscapes with tall and wide barriers. 
Moreover, D-Wave quantum annealers have only a low precision of approximately $10^{-2}$ for couplings' strengths and local fields' biases (i.e., entries of $J$ and $h$ in \eqref{Ising_def}), scaling $J$ and $h$ to a very limited range, (e.g., $[-2,2]$).
Several sources of noise such as thermal excitations and control errors contribute to the low precision of the machines; in \cite{albash2014}, Albash et al. propose a noise model for D-Wave devices that includes the control noise of the local field and couplings of the chip.  Zhu et al. \cite{zoshk2016} show that increasing the classical energy gap beyond the intrinsic noise level of the machine can improve the success of the D-Wave Two QA.
%,at the cost of producing considerably easier benchmark instances. 
In \cite{zoshk2016}, {\emph{resilience}} is defined as a metric for measuring the resistance of a problem to noise. More precisely, resilience is the probability that the ground state does not change under random field fluctuations found on the chip. In this paper, we argue that restricting the range of the coefficients of a problem could improve resilience; and our results, presented in Section \ref{sec:res} of this paper, support this argument

Studies in \cite{biasedGSsampling,GSstat} suggest that QAs could be biased in finding degenerate ground states; in other words, sampling with QAs does not return different degenerate solutions with uniform probability. Moreover, degenerate ground states are easier to reach; therefore, benchmarks created for QAs tend to avoid degeneracy and have a unique ground state \cite{katzgraber2015}. While there is a need to rigorously study whether highly degenerate low-energy excited states could impact reaching the ground state, avoiding degeneracy could be an alternative. Recall that in integer encoding, an integer variable, $x$, is substituted with $c^t y$, where we assume $c \in \mathbb Z_+^d$ and $y \in \mathbb B^d$. Any binary vector $y$ returns an integer value. If the total number of binary combinations, i.e., $2^d$, is larger than the summation of entries of $c$, i.e., $\sum_{i=1}^d c_i$, some of the integers occur at more than one binary combination. The \emph{code-word} of an integer value $\chi$, for $\chi \in \{1, 2, \ldots, \kappa\}$, refers to the cardinality of the set $\{ y: c^t y = \chi \}$. Code-words more than one is what we refer to as {\emph{redundancy}} in this paper. Under the assumption that the QA is inclined to observe integers with higher code-words, it is ideal in an integer encoding to have a unique code-word for all integers. Binary encoding, for example, keeps the code-word of each integer uniformly at one, but the code-word in unary encoding is highly variant and the redundancy has its peak at $\l \lceil \frac{\kappa}{2} \r \rceil$ with word count $\kappa \choose \lceil \frac{\kappa}{2}\rceil$.

This work presents an integer encoding with minimal width that creates noise-resilient Ising models. 
%Keeping the degree of an integer encoding small is desired because it defines the size of the Ising model. 
Since one or several qubits is assigned to each spin variable and the number of qubits on a chip is limited, it is desirable to keep the size of the Ising model, and thus the width of the encoding, as small as possible. Binary encoding has the minimum width and uniform code-word; however, the coefficients $c_j$'s in the binary encoding could get arbitrarily large and result in an error-prone (noisy, non-resilient) Ising model. On the other hand, unary encoding does not expand the range of the coefficients of the problem, but it has an exceedingly large degree and redundancy. 

The notation used in this paper is as follows. Generally, we reserve the upper-case letters, lower-case letters, and Greek alphabet for matrices, vectors, and scalars, respectively, with only a few exceptions; these exceptions, however, should be clear from the context. An all-ones vector and the identity matrix are denoted by $e$ and $I$, respectively. Entries of a vector or  matrix are differentiated with subscript indices. $\lceil \cdot \rceil$ and $\lfloor \cdot \rfloor$ denote the ceiling and floor of a number, respectively. The $\log(\cdot)$ function is in base two; and $\sg(\alpha)$ returns the sign of $\alpha$. The function $\diag(A): \mathbb R^{n\times n} \rightarrow \mathbb R^n$ returns a vector with the diagonal entries of $A$; and the function $\Diag(a): \mathbb R^n \rightarrow \mathbb R^{n\times n}$  returns a diagonal matrix with entries of $a$ on the diagonal.

This paper is organized as follows. In the next section, we present bounded-coefficient encoding under the assumption that an upper bound on the coefficients of the integer encoding is available. In Section \ref{sec:findingUB}, we present how, given the precision, we may find an upper bound on the coefficients of the integer encoding. We test our algorithm by comparing binary and bounded-coefficient encodings in Section \ref{sec:res}. Finally, we conclude our discussion in Section \ref{sec:conclusion}.

%%%%%%%%%%%%%%%%%%%%%%%%%%%%%%%%%%%%%%%
%%%%%%%%%%%%%%%%%%%%%%%%%%%%%%%%%%%%%%%
\section{Bounded-Coefficient Encoding}\label{sec:bounded-coeff encoding}
Let $x$ be an integer variable with upper bound $\kappa^{x}$, and $\mu^{x}$ be an upper bound on the coefficients of the integer encoding. In other words, we represent $x$ as
\begin{equation}\label{intEncoding_sec2}
x = c^t y,
 \end{equation}
 where $y\in \mathbb B^{d^{x}}$, $c \in \mathbb Z_+^{d^x}$, and $c_i \le \mu^x$ for $i =1, \ldots, d^x$. The encoding we propose in this paper, which we refer to as \emph{bounded-coefficient encoding}, is summarized in the following algorithm:

%%%
\begin{algorithm}
\caption{Bounded-Coefficient Encoding}
\label{alg:bounded-coeff}
\begin{tabbing} 
\quad \quad \= \quad \quad  \=\quad \=\quad \=\kill \\
\> Inputs:\\
\> \> $\kappa^x$: upper bound on the integer variable $x$\\
\> \> $\mu^x \ll \kappa^x$: upper bound on the coefficients of the encoding\\
\> {Output:}\\
\> \> $c^x$: integer encoding coefficients\\ 
\\
\> {\bf if} $\kappa^x < 2^{\lfloor \log (\mu^x) \rfloor +1}$ \\
\> \> {\bf return} 
\end{tabbing}
\beq
\label{bounded-coeff_case0} 
c^x= \l [2^0, 2^1, \ldots, 2^{\lfloor \log (\kappa^x) \rfloor-1}, \kappa^x -\sum_{i=1}^{ \lfloor \log (\kappa^x) \rfloor} 2^{i-1}  \r ]
\eeq
\begin{tabbing} 
\quad \quad \= \quad \quad  \=\quad \=\quad \=\kill \\
\> {\bf else}\\
\>\> {\bf compute} $\rho=\lfloor \log \mu^x \rfloor + 1 $, $\nu = \kappa^x - \sum_{i=1}^{\rho } 2^{i-1}$,  and $\eta = \left \lfloor \frac{ \nu} {\mu^x} \right \rfloor$\\
\>\> {\bf return} vector $c^x$ with entries 
\end{tabbing}
\normalfont
\begin{equation}
\label{bounded-coeff_eq}
c^x_i=\begin{cases}
2^{i-1} & \text{ for } i=1, \ldots , \rho\\
\mu^x & \text{ for } i=\rho+1, \ldots , \rho + \eta \\
\nu - \eta \mu^x & \text{ for } i=\rho + \eta +1 \text{ if } \nu - \eta \mu^x \neq 0
\end{cases} 
\end{equation}
\\
\end{algorithm}
%%%

Notice that if $\kappa^x < 2^{\lfloor \log \mu^x \rfloor +1}$, then the binary encoding respects the upper bound on the coefficients. These cases are captured in \eqref{bounded-coeff_case0}, and the width of encoding in these instances is $ \lfloor \log (\kappa^x) \rfloor +1$.

When $\kappa^x \ge 2^{\lfloor \log \mu^x \rfloor +1}$, the encoding is derived using \eqref{bounded-coeff_eq} and the width of the bounded-coefficient encoding is
\begin{equation}
d^x=\begin{cases} \rho+\eta+1 & \text{ if } \nu-\eta \mu \neq 0\\ \rho+\eta & \text{ otherwise } \end{cases}. 
\end{equation}

Here are a few demonstrative examples of bounded-coefficient encoding:
\begin{itemize}
\item $\kappa=12$ and $\mu = 8$, the bounded-coefficient encoding is $c=\l [1, \ 2, \ 4,\ 5\r ]^t;$ 
\item $\kappa=20$ and $\mu = 6$, the bounded-coefficient encoding is $c=\l [ 1,\ 2,\ 4,\ 6,\ 6,\ 1\r  ]^t.$
\end{itemize}

We assume $\kappa^x \gg 2^{\lfloor \log \mu^x \rfloor +1}$ when we refer to the bounded-coefficient encoding, and our general statements are focused on these cases where we have multiple coefficients of size $\mu^x$. Propositions \ref{validity} and \ref{minDegree}, also, refer to these cases described by \eqref{bounded-coeff_eq}; these propositions hold true for binary encoding; hence, they hold for the cases that are derived with \eqref{bounded-coeff_case0}, such as our first example above. 

%%%
\begin{defn} 
\label{kComplete}
An encoding $c$ is called $\kappa$-complete if it can encode only and all integers $\{0,1,\ldots, \kappa\}$. 
\end{defn}
\begin{prop}\label{validity}
The bounded-coefficient encoding, generated by Algorithm \ref{alg:bounded-coeff}, is $\kappa^x$-complete. 
\end{prop}
By the first part of our encoding, i.e., $1, \ldots, 2^{\rho-1}$, we can encode all integers $\{ 0, 1, \ldots, \mu^x\}$. By adding $\mu^x$ to those binary combinations, we can generate all integers from $\mu^x$ to $2\mu^x$; if we add two $\mu^x$ factors, we get integers $2\mu^x$ to $3\mu^x$, and so on. Finally, we are guaranteed that we get integers $\eta\mu^x$ to $\kappa^x$ because $\kappa^x- \eta\mu^x < \sum_{i=1}^{\rho} 2^{i-1} - \mu^x$ by the fact that $\nu - \eta \mu^x < \mu^x$.\\ 

%%%%%
\begin{defn}
A sub-encoding $\tilde c$ of an encoding $c$ is a choice of entries of $c$, with a fixed ordering arranged in a column vector $\tilde c$. 
\end{defn}
\begin{prop}\label{minDegree}
Bounded-coefficient encoding, generated by Algorithm \ref{alg:bounded-coeff}, has minimum width among all $\kappa^x$-complete integer encodings of the form \eqref{intEncoding} that respect an upper bound on their coefficients.
\end{prop}
\begin{proof}
Our arguments rely on the optimality of binary encoding, or, more rigorously, the following points: 
\begin{enumerate}[(a)]
\item \label{binaryEncdPoint1} binary encoding has minimum width; in other words, the binary encoding of values  $\{ 0, 1, \ldots, \mu^x \}$ has width $d_B = \lfloor \log \mu^x \rfloor +1 $, and no other $\mu^x$-complete integer encoding could have a width less than $d_B$.
\item \label{binaryEncdPoint2}  binary encoding of width $d_B$, i.e., $\l(2^0, \ 2^1, \ 2^2, \ \ldots, \ 2^{d_B-1} \r)$, encodes integers of maximum value $2^{d_B} -1$; i.e., $\sum_{i=0}^{d_B-1} 2^i = 2^{d_B } -1$. Notice that by \eqref{binaryEncdPoint1}, this is the largest number that can have an encoding of width $d_B$.
\end{enumerate}

It is worth mentioning that if $\kappa^x < 2^{\lfloor \log \mu^x \rfloor +1}$, then the bounded-coefficient encoding has width $\lfloor \log \kappa^x \rfloor + 1$, and by \eqref{binaryEncdPoint1} has minimum width. 

For more-general cases, our proof is by contradiction. Suppose $c$ is the bounded-coefficient encoding derived by Algorithm \ref{alg:bounded-coeff}, and $f$ is a bounded-coefficient encoding of smaller width; i.e., $d_c> d_f$. 
Let $f_\mu$ be the minimal sub-encoding of $f$ that encodes $\{ 0, 1, \ldots, \mu^x\}$, and let $d_{f_\mu} \le d_f$ be the width of $f_\mu$. Similarly, $c_\mu$ in the bounded-coefficient encoding derived by Algorithm \ref{alg:bounded-coeff} refers to the sub-encoding of $c$ required to encode integers less than or equal to $\mu^x$. As implied by Algorithm \ref{alg:bounded-coeff} and \eqref{bounded-coeff_eq}, $d_{c_\mu} = \rho$. 
By \eqref{binaryEncdPoint1}, we conclude $d_{f_\mu} \ge \rho$. We can, now, consider two cases: $d_{f_\mu} = \rho$ and  $d_{f_\mu} > \rho$.

Consider the case where $d_{f_\mu} = \rho$. By \eqref{binaryEncdPoint1} and \eqref{binaryEncdPoint2}, $\sum_{\alpha \in c_\mu} \alpha \ge \sum_{\alpha \in f_\mu} \alpha$; therefore, $\kappa^x -  \sum_{\alpha \in f_\mu} \alpha \ge \nu$, where $\nu$ is as defined in Algorithm \ref{alg:bounded-coeff}. Since all of the coefficients need to be bounded by $\mu^x$, and by the fact that $\frac{\kappa^x -  \sum_{\alpha \in f_\mu^x} \alpha }{\mu^x} \ge \frac{\nu}{\mu^x}$, we conclude that having $d_f < d_c$ is impossible. 

Consider now the case that $d_{f_\mu} \ge \rho + 1$. As all of our coefficients are bounded by $\mu^x$, $\sum_{\alpha \in f_\mu} \alpha < 2\mu^x$; otherwise, this contradicts that $f_\mu$ is the minimal sub-encoding required to encode integers less than or equal to $\mu^x$. While $\sum_{\alpha \in f_\mu} \alpha < 2\mu$, $\sum_{\alpha \in c_\mu } \alpha  + \mu \ge 2\mu$. Similar to the previous case, we conclude that $\kappa^x - \sum_{\alpha \in f_\mu} \alpha > \kappa^x - \sum_{\alpha \in c_\mu } \alpha  + \mu^x$; hence, $d_f < d_c$ would not be possible.
\qed
\end{proof}
%%%%%
As mentioned earlier, a potentially advantageous property of an integer encoding is uniform or low-variant redundancy. In general, the bounded-coefficient encoding has redundancy because we have more than one coefficient with a value of $\mu$. It is, however, worth mentioning that forcing the following constraints on the binary variables of \eqref{intEncoding_sec2} makes the word count of each integer value unique: 
\begin{equation}
\begin{array}{rcll}
\sum_{j=1}^{\rho} c_j y_j & \ge & (2^{\rho} - \mu^x) y_i  & \text{for }  i=\rho +1, \\
y_{i} &\ge& y_{i+1} &  \text{for } i=\rho+1, \ldots, d^x-1,\\
\sum_{j=1}^{\rho} c_j y_j &\ge& (2^{\rho} - c_{d^x}) y_{d^x}.
\end{array}
\end{equation}

Note that $(2^{\rho} - \mu^x) \le \mu^x$; however, $(2^{\rho} - c_{d^x})$ may not necessarily be less than or equal to $\mu^x$. If we wish to limit the coefficients of our constraints to $\mu^x$, we can substitute $(2^{\rho} - c_{d^x}) y_{d^x}$ in the right-hand side of the last inequality with $\mu^x y_{d^x-1}y_{d^x} + (2^{\rho} - \mu^x - c_{d^x}) y_{d^x}$ whenever $2^{\rho} - c_{d^x} > \mu^x$.  Handling constraints, despite being possible, is non-trivial for QA; it may introduce exceedingly large coefficients (see discussion in \cite{karimi_GD}), and could be contradictory to the purpose of this work. 

As we presented earlier in \eqref{binToSpin} and \eqref{int2SpinEncoding}, we may directly encode integer variables into spin variables. Moreover, by the fact that in the bounded-coefficient encoding $\sum_{j=1}^{d^{x}} c_j = \kappa^x$, the encoding of each integer variable into spin variables would be
\begin{equation}
\label{intToSpin_BC}
x = \frac{1}{2} \l ( \kappa^x + \sum_{j=1}^{d^{x}} c_j^{x} s_j^{x} \r). 
\end{equation}
In the next section, where we try to find the upper bound on the coefficients of the encoding, i.e., $\mu^x$, we use the integer-to-spin transformation. 

%%%%%%%%%%
\section{\large Finding the Upper Bound on the Coefficients of the Encoding}
\label{sec:findingUB}
We explained earlier, in Section \ref{sec:intro}, that the QA scales the coefficients of an Ising model to a limited range and has low precision. We also mentioned that resilience to noise could be helpful in overcoming the issue of precision. We propose that restricting the range of the coefficients of the problem could improve resilience. To this end, we wish to find upper bounds on the coefficients of the encoding, i.e., $\mu$, such that after the encoding of the integer problem \eqref{UIQP} to an Ising model \eqref{Ising_def}, the ratio of the smallest local field in magnitude to the largest one exceeds a threshold, i.e., 
\begin{equation}
\label{localFieldsRatio}
\frac{ \min_i |h_i|}{\max_i |h_i|} \ge \eps_l.
\end{equation}
Similarly, we wish to bound the ratio of the smallest coupler to the largest one:
\begin{equation}
\label{couplerRatio}
\frac{ \min_{i,j} |J_{ij} | }{\max_{i,j} |J_{ij} | } \ge \eps_c.
\end{equation}

Alternatively, one may desire to have the local fields biases and couplings' strengths well-separated to avoid the effect of noise, or in other words, having 
\begin{equation*}
\min_{i, j} \l | h_i - h_j \r | \ge \eps_l, 
\end{equation*}
and 
\begin{equation*}
\min_{\substack{i,j \\ l,k}} \l | J_{ij} - J_{lk} \r | \ge \eps_c,
\end{equation*}
could be desirable. In this paper, however, we focus merely on inequalities \eqref{localFieldsRatio} and \eqref{couplerRatio}.

Suppose we have a quadratic function 
\beq
\label{UIQP_obj}
f(x)= x^t Q x + q^t x,
\eeq
where $Q$ is symmetric and the domain of $f(x)$ is  $x=\l [x_1, x_2, \ldots, x_n \r]^t$ for $x_i \in \{0, \ 1, \ \ldots, \ \kappa^{x_i}\}$. Note that this is the function which we aim to minimize in \eqref{UIQP}. Also, note that since \mbox{$x_i x_j = x_j x_i$} for two integers $x_i$ and $x_j$, we may substitute $Q$ with $\frac{1}{2} (Q +Q^t)$ when $Q$ is not symmetric. Of course, everything presented here is under the assumption that the coefficients of the problem (before any encoding) respect inequalities \eqref{localFieldsRatio} and \eqref{couplerRatio}, i.e., unary encoding satisfies them. Let us define the encoding matrix denoted by $C$ as
\beq
\label{encdMatrix}
C= 
\begin{bmatrix}
c^{x_1}_1& c^{x_1}_2 & \dots & c^{x_1}_{d^{x_1}} & 0 & 0 & \dots & 0 & 0 &  \dots & \dots & \dots & 0 & 0& 0 & \dots & 0 \\ 
0 & 0 & \dots & 0& c^{x_2}_1& c^{x_2}_2 & \dots & c^{x_2}_{d^{x_2}}  & 0& \dots & \dots & \dots & 0 & 0& 0 & \dots & 0 \\ 
\vdots & \vdots & \ddots & \vdots & \vdots & \vdots & \ddots & \vdots  & \vdots & \ddots & \ddots & \ddots & \vdots & \vdots & \vdots & \ddots & \vdots \\ 
0 & 0 & \dots & 0& 0 & 0 & \dots & 0& 0& \dots & \dots & \dots & 0 & c^{x_n}_1& c^{x_n}_2 & \dots & c^{x_n}_{d^{x_n}}  
\end{bmatrix}, 
\eeq
and the vector of spin variables as
\beq
s = \begin{bmatrix}
s^{x_1}_1& s^{x_1}_2 & \dots & s^{x_1}_{d^{x_1}} &s^{x_2}_1& s^{x_2}_2 & \dots & s^{x_2}_{d^{x_2}} & \dots & \dots & \dots & s^{x_n}_1& s^{x_n}_2 & \dots & s^{x_n}_{d^{x_n}}
\end{bmatrix}^t. 
\eeq
Using \eqref{intToSpin_BC}, we conclude that the encoding of integer variables, $x$,  into spin variables, $s$, refers to the following substitution:
\beq
\label{intToSpin_vecForm}
x = \frac{1}{2} \l ( \kappa + C s \r), 
\eeq
where $\kappa = \l[\kappa^{x_1}, \ \kappa^{x_2}, \ \ldots, \ \kappa^{x_n} \r]^t $ is the vector of all of the upper bounds on the integer variables. Using this substitution in \eqref{UIQP_obj}, we get
\beqa
\label{IsingObjFun}
f(s) = \frac{1}{4} s^t \l ( C^t Q C - \Diag ( \diag(C^t Q C)) \r) s & + & \frac{1}{2} \l( C^t Q \kappa + C^t q \r)^t s \notag \\
&+& \frac{1}{4} \l( \kappa ^t Q \kappa + e^t  \Diag ( \diag(C^t Q C)) e + 2 q^t \kappa \r).
\eeqa
We are excluding the diagonal entries of $C^t QC$ in the quadratic term because they correspond to the square of spin variables, i.e., $s_i^2$, and $s_i^2$ is a constant equal to one.  The local field bias corresponding to each spin variable and the coupling strength corresponding to each pair of spin variables, derived in \eqref{IsingObjFun}, is summarized below:
\beqa
s_j^{x_i} & : & \frac{1}{2} \l [ Q \kappa + q  \r ]_i c_j^{x_i} =  \frac{1}{2} \l(q_i + \sum_{k=1}^{n} Q_{ik} \kappa_k \r) c_j^{x_i}, \label{hint0} \\
s_k^{x_i}s_l^{x_i} & : & \frac{ \l(Q_{ii} c_k^{x_i} c_l^{x_i}\r) }{2} \quad  \text{for } \ k,l \in \{1, \ldots, d^{x_i}\} \  \text{ and } \ k<l,  \label{hint1}\\
s_k^{x_i}s_l^{x_j} & :  & \frac{ \l(Q_{ij} c_k^{x_i} c_l^{x_j}\r) }{2}  \quad \text{for }\  k \in \{1, \ldots, d^{x_i}\} ,\ l \in \{1, \ldots, d^{x_j}\}, \  \text{ and } \ i,j \in \{1, \ldots, n\} : i<j, \label{hint2} \ \ \   
\eeqa
where $\l [ Q \kappa + q  \r ]_i$ in the first line denote the $i$-th entry of vector $Q \kappa + q$, as expected, and \eqref{hint1} and \eqref{hint2} incorporate the fact that $s_k^{x_i}s_l^{x_i} = s_l^{x_i}s_k^{x_i}$ and $s_k^{x_i}s_l^{x_j}= s_l^{x_j}s_k^{x_i}$. 

Since all of the local fields biases and couplings' strengths share the $\frac{1}{2}$ factor and after rescaling our concern is their ratio, we will drop the $\frac{1}{2}$ factor for our calculations in the rest of this section. 
From the previous section and the derivation of the bounded-coefficient encoding, we conclude that the smallest coefficient used in the encoding is 1; hence, the minimum absolute value of the local fields biases is 
\beq \notag
m_l = \min_{i} \left \{ \ \l | \l [ Q \kappa + q  \r ]_i  \r | \ \right \}, 
\eeq
%%
%\beq
%m_l = \min_{i} \left \{ \l | {q_i } + \sum_{\substack{j=1}}^n  {Q_{ij} \kappa_j}  \r | \right \}, 
%\eeq
and the minimum absolute value of couplings' strengths is 
\beq \notag
m_c = \min_{i,j} \left \{ \ \l | Q_{ii} \r| , \l| Q_{ij} \r| \ \right \}.
\eeq
The maximum absolute value of the local fields biases among all spin variables corresponding to an integer variable $x_i$ occurs at $\mu^{x_i}$---to be found---for each $i\in \{1,\ldots, n\}$. Therefore, by \eqref{localFieldsRatio}, we wish to have: 
\beqa
\label{LFratio}
\frac{m_l} { \l | \ \l [ Q \kappa + q  \r ]_i \ \r| \mu^{x_i} } \ge \eps_l.
\eeqa
Similarly, \eqref{couplerRatio} enforces the following conditions for the couplers: 
\beq
\label{Cratio}
\frac{m_c} { |Q_{ii}| (\mu^{x_i})^2 } \ge \eps_c, \ \text{and} \ \frac{m_c} { |Q_{ij}| \mu^{x_i}\mu^{x_j} } \ge \eps_c.
\eeq
Using \eqref{LFratio} and \eqref{Cratio}, we wish to obtain $\mu^{x_i}$'s that satisfy the following set of inequalities:
\beqa
\label{muIneqList}
\mu^{x_i} & \le& \frac{m_l} { \l |  \l [ Q \kappa + q  \r ]_i   \r| \eps_l } \label{muIneqList_1} \\
\mu^{x_i} &\le & \sqrt{\frac{m_c}{|Q_{ii} |\eps_c} } , \label{muIneqList_2} \\
\mu^{x_i} \mu^{x_j} &\le & \frac{m_c}{|Q_{ij}| \eps_c}.\label{muIneqList_3}
\eeqa

We may solve a feasibility problem to find a solution to the above set of inequalities, i.e., an optimization or auxiliary objective function along with these inequalities. A well-justified objective function could be $\max \  \min_i \{ \mu^{x_i} \}$. Solving such problems is normally costly because \eqref{muIneqList_3} is non-convex and $\mu^{x_i}$'s are required to be integers. Alternatively, we present an algorithm below that heuristically finds $\mu^{x_i}$'s.

Note that any set of $\mu^{x_i}$'s satisfying the above inequalities will guarantee that \eqref{localFieldsRatio} and \eqref{couplerRatio} hold. 
Our proposed algorithm for finding $\mu^{x_i}$'s first initializes $\mu^{x_i}$'s using inequalities \eqref{muIneqList_1} and \eqref{muIneqList_2}. If \eqref{muIneqList_3} is satisfied for all $i$ and $j$ at this step, it terminates. Otherwise, it greedily decrease $\mu^{x_i}$ or $\mu^{x_j}$ for an $i$ and $j$ pair for which the failure happened. In this process, we take $\frac{\kappa^{x_i} }{\mu^{x_i}}$ as an estimate on the width of the encoding, so when we want to decrease either $\mu^{x_i}$ or $\mu^{x_j}$, we choose the one that gives a lower combined width.  See Algorithm \ref{alg:Find_$mu$_Ising} for a formal presentation of this algorithm.

\begin{algorithm}
\caption{Finding the Upper Bounds on the Coefficients of the Encoding}
\label{alg:Find_$mu$_Ising}
\begin{tabbing} 
\quad \quad \= \quad \quad  \=\quad \=\quad \=\kill \\
Inputs:\\
\> \normalfont $\kappa, \ q,\ Q, \ \eps_l, \ \eps_c$\\
\> {\bf {compute}}  $Q\kappa + q$\\
\> {\bf{set}} $m_l =  \min_{i} \left \{ \l | \  \l[Q\kappa + q \r]_i \ \r | \right \}$,   and $m_c = \min_{i,j} \left \{ |Q_{ii}| , |Q_{ij}|\right \} $ \\
{Output:}\\
\> \normalfont  $\mu^{x_i}$ for $i= 1, 2, \ldots , n$ \\ 
\\
\normalfont {\bf{initialize}} $\mu^{x_i}= \l \lfloor \min \left \{  \frac{m_l} { \l | \ \l [ Q \kappa + q  \r ]_i  \ \r| \eps_l }, \ \sqrt{\frac{m_c}{|Q_{ii}| \eps_c} } \right \} \r \rfloor$\\
\\
{\bf while } \normalfont any $ \l(\mu^{x_i} \mu^{x_j} > \frac{m_c}{|Q_{ij}| \eps_c} \r) $\\
\\
\> \normalfont let $i,j = \arg \max_{i,j} \left \{ \mu^{x_i} \mu^{x_j} - \frac{m_c}{|Q_{ij}| \eps_c} \right \} $\\
\\
\> \normalfont let $\xi_i = \frac{\kappa^{x_i}}{ \mu^{x_i} -1} + \frac{\kappa^{x_j}}{ \mu^{x_j}}$ and  $\xi_j = \frac{\kappa^{x_i}}{ \mu^{x_i} } + \frac{\kappa^{x_j}}{ \mu^{x_j} -1}$\\
\\
\> \bf {if} $\xi_i < \xi_j$\\
\>\> $\mu^{x_i} =\mu^{x_i}-1$\\
\> \bf{else}\\
\>\> $\mu^{x_j} =\mu^{x_j}-1$	\\
\end{tabbing} 
\end{algorithm}

It is worth mentioning that if we wish to keep $\mu^{x_i}$'s for all integer variables equal, we may use the value
\beq
\mu = \min_i \ \{ \mu^{x_i} \}.
\eeq
Moreover, if in \eqref{UIQP_obj} the quadratic terms were missing, i.e., $f(x)= \ q^t x$, we would get 
\beq
f(s) = \frac{1}{2} q^t \l ( \kappa + Cs\r) = \frac{1}{2} q^t \kappa + \frac{1}{2} \l( C^t q \r)^t s.
\eeq
Ignoring the $\frac{1}{2}$ factor (since it is common for all spin variables), we get $m_l = \min_i \ \{ \l |q_i \r| \}$, and our ratio condition reduces to 
\beq
\frac{m_l} { |q_i| \mu ^{x_i} } \ge \eps_l;
\eeq
therefore,
\beq
 \mu ^{x_i} = \l\lfloor \frac{m_l} { |q_i| \eps_l} \r\rfloor.
\eeq

Finally, we would like to point out that unlike QA, in several heuristic methods, variables tend to be binary, i.e., in the $\{0,1\}$ domain referred to as binary. In Appendix \ref{sec:intToBin}, we present a modification of the method in this section that could be employed in this case.

%%%%%%%%%%%%%%
\section{Numerical Experiment}
\label{sec:res}
In this section, we test the bounded-coefficient encoding. We compare binary and bounded-coefficient encodings on ten randomly generated instances. To diversify our instances, however, we use several procedures to generate them. Let us define the set $U_\alpha = \{ 0, \ \pm 1, \ \pm 2, \ldots, \ \pm \alpha\}$. In all of our instances, we have five integer variables, i.e., $n=5$; the upper-bound on all of the integer variables is 50, i.e., $\kappa^{x_i} = 50$ for $i = 1, \ldots, 5$; and the matrix $Q$ in the quadratic term has a sparsity around 50\%. In half of our instances, the generated $Q$ is positive definite, and in the rest it is indefinite; these two categories are referred to as convex and non-convex instances, respectively. 

For the convex instances, entries of matrix $Q$ are initially drawn from $U_2$; then, $\lambda I$ is added to $Q$ to make $Q$ positive definite. Our choice of $\lambda$ is $\l \lceil | \min \{ \lambda_{\min}, 0 \} | + r \r \rceil$, where $\lambda_{\min}$ is the minimum eigenvalue of $Q$ and $r$ is a random number between 0 and 1. A feasible sparse integer vector, i.e., $x^\ast$, where $x^\ast_i \in \{0,1,\ldots, 50\}$ if $x^\ast_i \neq 0$, is then generated, and we set $q = -2Qx^\ast$. Note that by the fact that $Q$ is positive definite and by our choice of $q$, $x^\ast$ is the unique optimal solution to the generated instance of \eqref{UIQP}. These instances are shown with the prefix {\sf convex} in Tables \ref{table:res_BC} and \ref{table:res_binary}.
In the non-convex instances, entries of $Q$ and $q$ are from separate $U_\alpha$'s. We denote these instances with pairs $(U_Q, U_q)$, where $U_Q$ and $U_q$ are distributions for the quadratic terms ($Q$) and linear terms ($q$), respectively. Our data sets are $(U_2, U_{200}),\ (U_5, U_{200}), \ (U_5, U_{10}), \ (U_5, U_{100}), \ $ and $(U_{10}, U_0)$.

As mentioned earlier, we take {\emph{resilience}} as the measure of success. More specifically, we hypothesized that the bounded-coefficient encoding is a technique for representing a UIQP problem as a UBQP problem that is more robust against noise, and resilience directly measures the robustness of an instance of UBQP against noise. In \cite{zoshk2016}, the resilience $R$ of an instance is defined as
\beq
R= \frac{n_{\text{same} } } { n_{\text{trial} } }, 
\eeq
where $n_{\text{same}}$ is the number of times, among all $n_{\text{trial}}$, that the original ground state does not change with random noise perturbations. In other words, $n_{\text{trial} }$ different noise matrices with entries drawn from the normal distribution $\mathcal N(0, \eps)$ are generated; each  is added to a scaled Ising model and the perturbed Ising model is solved exactly. $n_{\text{same}}$,  then, refers to the number of times that the perturbed Ising model returns the same ground state as the unperturbed one.  
The number of trials, $n_{\text{trial} }$, in our experiment is set to 10. We take $\eps_l$ and  $\eps_c$ to be 0.01, and we scale Ising models to $J \in [-1, 1]$ before adding the noise. After finding the upper bound on the coefficients of the encoding, with Algorithm \ref{alg:Find_$mu$_Ising} and $\eps_l = \eps_c = 0.01$, we obtain $\text{Ising}_{\text{bounded}}$ through the bounded-coefficient encoding. Also, by using the binary encoding, we derive $\text{Ising}_{\text{binary}}$. Then, we measure the resilience of each of these Ising models at $\eps \in \{ 0.001, 0.002, \ldots, 0.01\}$. The resilience of $\text{Ising}_{\text{bounded}}$ and $\text{Ising}_{\text{binary}}$ are summarized in Tables \ref{table:res_BC} and \ref{table:res_binary}, respectively. The bounded-coefficient encoding has significantly outperformed the binary encoding; $\text{Ising}_{\text{bounded}}$ is five times more resilient to noise than $\text{Ising}_{\text{binary}}$, on average. Also, in most of our test cases, the convex instances and $(U_{10}, U_0)$, the resilience is almost zero with the binary encoding, except for a few exceptions at $\eps=0.001$, whereas for bounded-coefficient encoding, the resilience stays above zero and is considerably large for $\eps \le 0.005$. In order to visualize the difference between the two encodings, we plot the average resilience (over all instances) with respect to $\eps$, the standard deviation of the added noise, in Figure \ref{fig:res_avg}. It is obvious from this plot that the resilience with binary encoding stays marginally above 0, but for the bounded-coefficient encoding it decreases at a slower rate and remains considerably above 0.

\begin{table} \caption{Resilience with Bounded-Coefficient Encoding.}
\begin{center}
\begin{threeparttable}
\begin{tabular}{|c|cccccccccc|}
\hline
\diagbox{Ins.}{$\eps$} & 0.001 & 0.002 & 0.003 & 0.004 & 0.005 & 0.006 & 0.007 & 0.008 & 0.009 & 0.01 \\
\hline 
\sf{convex-1} &1.0 &   1.0 &   0.6 &  0.4 &  0.8 &  0.5 & 0.2 &  0.1 & 0.1&   0.0\\
\sf{convex-2} &0.7&0.5&0.4&0.1&0.2&0.0&0.0&0.2&0.0&0.0\\
\sf{convex-3} &1.0&0.7&0.6&0.5&0.2&0.1&0.0&0.1&0.0&0.0\\
\sf{convex-4} &1.0&0.9&0.6&0.3&0.2&0.1&0.1&0.0&0.1&0.0\\
\sf{convex-5} &1.0&1.0&0.5&0.6&0.3&0.4&0.2&0.2&0.4&0.1\\
$(U_2, U_{200})$ &1.0&1.0&1.0&1.0&0.9&1.0&0.9&0.8&1.0&0.8\\
$(U_5, U_{200})$ &1.0&0.9&0.9&0.9&0.8&0.4&0.7&0.8&0.5&0.6\\
$(U_{10}, U_0)$ &1.0&0.8&0.4&0.3&0.4&0.1&0.4&0.2&0.1&0.2\\
$(U_5, U_{10})$ &1.0&1.0&1.0&0.9&0.7&0.9&0.4&0.4&0.6&0.7\\
$(U_5, U_{100})$ &1.0&1.0&1.0&1.0&1.0&1.0&1.0&0.7&0.7&0.6\\
\hline
average &  0.97  &0.88&  0.7&   0.6&   0.55&  0.45&  0.39&  0.35&  0.35&  0.3 \\
\hline
\end{tabular} 
\label{table:res_BC}
\end{threeparttable}
\end{center}
\end{table}

%%%
\begin{table} \caption{Resilience with Binary Encoding.}
\begin{center}
\begin{threeparttable}
\begin{tabular}{|c|cccccccccc|}
\hline
\diagbox{Ins.}{$\eps$} & 0.001 & 0.002 & 0.003 & 0.004 & 0.005 & 0.006 & 0.007 & 0.008 & 0.009 & 0.01 \\
\hline 
\sf{convex-1} &0.1&0.0&0.0&0.0&0.0&0.0&0.0&0.0&0.0&0.0\\
\sf{convex-2} &0.0&0.0&0.0&0.0&0.0&0.0&0.0&0.0&0.0&0.0\\
\sf{convex-3} &0.1&0.0&0.0&0.0&0.0&0.1&0.0&0.0&0.0&0.0\\
\sf{convex-4} &0.0&0.0&0.0&0.0&0.0&0.0&0.0&0.0&0.0&0.0\\
\sf{convex-5} &0.1&0.0&0.0&0.1&0.0&0.0&0.0&0.0&0.0&0.0\\
$(U_2, U_{200})$  &0.6&0.4&0.2&0.1&0.3&0.1&0.3&0.1&0.4&0.2\\
$(U_5, U_{200})$  &0.5&0.7&0.3&0.4&0.5&0.3&0.2&0.5&0.2&0.2\\
$(U_{10}, U_0)$  &0.4&0.0&0.0&0.0&0.0&0.0&0.0&0.0&0.0&0.0\\
$(U_5, U_{10})$  &0.3&0.1&0.1&0.1&0.1&0.0&0.1&0.1&0.1&0.0\\
$(U_5, U_{100})$ &0.5&0.3&0.1&0.4&0.3&0.4&0.3&0.1&0.4&0.1\\
\hline
average &0.26 & 0.15 & 0.07 & 0.11 & 0.12 & 0.09 & 0.09 & 0.08 & 0.11 & 0.05\\
\hline
\end{tabular} 
\label{table:res_binary}
\end{threeparttable}
\end{center}
\end{table}
%%%%

\begin{figure}
 \begin{center}
  \caption{Resilience Averaged over All of the Instances.}
    \includegraphics[width=10cm]{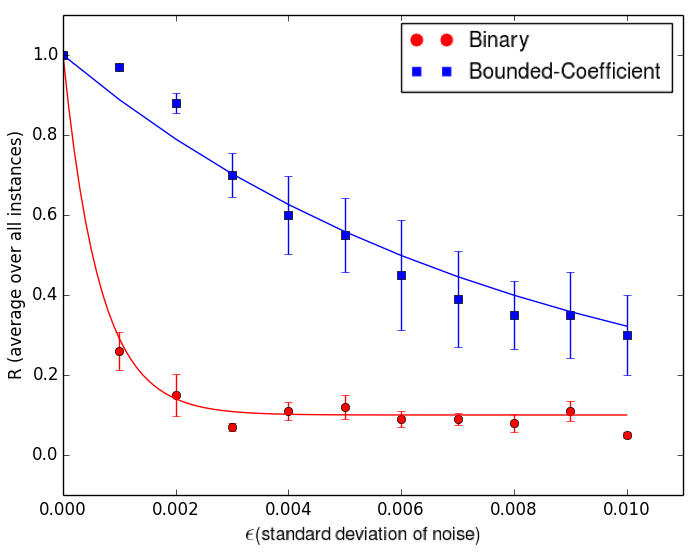}
    \label{fig:res_avg}
 \end{center}   
\end{figure}

Note that in \cite{zoshk2016}, it is suggested to set $\eps = \frac{2}{J_{\max} }$, where $J_{\max}$ is the largest $J_{ij}$. Here, our approach slightly differs. We assume that we know the precision of the machine, and we derive the upper bound on the coefficients of our encoding to accommodate this restriction. The precision of the current D-Wave machine is in order of $10^{-2}$, which motivates our choice of $\eps = 10^{-2}$. 
%%%%%%%%%%%%%%
\section{Conclusion and Discussion}
\label{sec:conclusion} 
In this paper, we presented an encoding to represent an unconstrained integer quadratic programming problem as an Ising model. To deal with the low precision of quantum annealers, we suggested bounding the value of the coefficients in the encoding. This restricts the range of the local fields biases and couplings' strengths in the derived Ising model, thereby creating Ising models that are more robust against noise after scaling. Resilience is used as a metric for the robustness of the Ising models. We compared bounded-coefficient encoding with the binary encoding. We infer from our results that bounding the coefficients of the encoding, as in the bounded-coefficient encoding, significantly improves the resilience of the model. The drawbacks are that the size of the derived Ising model is larger and it introduces redundancy. 

In our proposed technique, we forced the ratios of the minimum absolute value of the local fields biases and the couplings' strengths to their respective maximums exceed a certain threshold tolerance that is correlated or equal to the quantum annealer's precision. However, all of our calculations are embedding-free, i.e., they ignore the Chimera graph structure. In theory, having an embedding with equal chain length for all variables and equal chain connectivity for any present quadratic terms will respect our calculation; however, finding such an embedding is nontrivial. Notice that, after the encoding, all spin variables corresponding to an integer variable are connected; additionally, all spin variables corresponding to $x_i$ and $x_j$ are connected if the term $x_ix_j$ appears in the integer formulation; so, the underlying graph of the Ising model is quite dense. Even for our small instances, the resulting Ising models either exceed the size that can be embedded on the current chip, or the chains' lengths differ in orders of magnitude (e.g., minimum and maximum chain lengths in orders 1 and 10, respectively). Our algorithm can be tested on the future generation of the quantum annealers with improved connectivity. 

According to our results in the previous section, using $\eps_l = \eps_c = 0.01$ to find the upper bound on the coefficients of the encoding, the derived Ising model stays robust against noise, with an average resilience above 0.5, and minimum resilience above $0.1$ for $\eps= 0.005$, i.e., $\frac{\eps_c}{2}$. This could suggest using $2 \eps_{QA}$ as $\eps_l$ and $\eps_c$ for finding the upper bounds on the coefficients of the encoding, where $\eps_{QA}$ is the precision of the quantum annealer. This, however, requires more experimentation and could be an avenue for future study. 

%%%%%%%%%%%%%%
%%%%%%
\appendix
\section{(Integer-to-Binary Encoding)}
\label{sec:intToBin}
In this section, we aim to find the upper bounds on the coefficients of the integer encoding when we reduce an unconstrained integer quadratic programming (UIQP) problem to an unconstrained binary quadratic programming (UBQP) problem in the $\{ 0,1\}$ domain.
Similar to what we discussed earlier, we assume that our UIQP problem has the form
\beq
\notag
\begin{array}{lll}
\min &  x^T Q x + q^t x,\\
\text{s.t.} &x_i \in \{0, 1, 2, \ldots, \kappa^{x_i} \} & \quad   \text{for }  i=\{1,2, \ldots, n\},  \nonumber \\ 
\end{array}
\eeq
with $Q$ being symmetric. 

We aim to represent the above problem as $f(y)$ by substituting $x= Cy$, where $C$ is the encoding matrix we had earlier, i.e., 
\beq
C= 
\begin{bmatrix}
c^{x_1}_1& c^{x_1}_2 & \dots & c^{x_1}_{d^{x_1}} & 0 & 0 & \dots & 0 & 0 &  \dots & \dots & \dots & 0 & 0& 0 & \dots & 0 \\ 
0 & 0 & \dots & 0& c^{x_2}_1& c^{x_2}_2 & \dots & c^{x_2}_{d^{x_2}}  & 0& \dots & \dots & \dots & 0 & 0& 0 & \dots & 0 \\ 
\vdots & \vdots & \ddots & \vdots & \vdots & \vdots & \ddots & \vdots  & \vdots & \ddots & \ddots & \ddots & \vdots & \vdots & \vdots & \ddots & \vdots \\ 
0 & 0 & \dots & 0& 0 & 0 & \dots & 0& 0& \dots & \dots & \dots & 0 & c^{x_n}_1& c^{x_n}_2 & \dots & c^{x_n}_{d^{x_n}}  
\end{bmatrix},
\eeq
and
\beq
y = \begin{bmatrix}
y^{x_1}_1& y^{x_1}_2 & \dots & y^{x_1}_{d^{x_1}} &y^{x_2}_1& y^{x_2}_2 & \dots & y^{x_2}_{d^{x_2}} & \dots & \dots & \dots & y^{x_n}_1& y^{x_n}_2 & \dots & y^{x_n}_{d^{x_n}}
\end{bmatrix}^t \in \mathbb B^{\sum_{i=1}^n d^{x_i} }. 
\eeq

After the substitution for $x$, we get the following equivalent binary formulation:
%%%
\beq
f(y) = y^t (C^t Q C) y + (C^tq)^t y.
\eeq
%%%

Unlike the spin variables for which the diagonal of $C^t Q C$ became constant, the diagonal of $C^t Q C$ is added to the linear term in this scenario since $y_i^2 = y_i$ for $y_i \in \{0,1\}$. Alternatively, we can represent $f(y)$ as 
%%%
\beq
f_B  = y^t Q_B y, 
\eeq
%%%
where 
%%%
\beq
Q_B = \l[ 
\ba{c;{2pt/2pt}c;{2pt/2pt}c;{2pt/2pt}c}
Q_{11} [c^{x_1}][c^{x_1}]^t + \Delta(q_1[c^{x_1}] )  & Q_{12} [c^{x_1}][c^{x_2}]^t & \ \  \cdots \ \ & Q_{1n} [c^{x_1}][c^{x_n}]^t\\
\\
Q_{21} [c^{x_2}][c^{x_1}]^t  & Q_{22} [c^{x_2}][c^{x_2}]^t + \Delta(q_2[c^{x_2}]) & \ \ \cdots \ \ & Q_{2n} [c^{x_2}][c^{x_n}]^t\\
\\
\vdots & \vdots & \vdots & \vdots \\
\\
Q_{n1} [c^{x_n}][c^{x_1}]^t  & Q_{n2} [c^{x_n}][c^{x_2}]^t & \ \ \cdots \ \ & Q_{nn} [c^{x_n}][c^{x_n}]^t + \Delta(q_n[c^{x_n}])
\ea
\r].
\eeq
%%%
The diagonal entries of $Q_B$ are the linear terms, i.e., coefficients of the variables $y_j^{x_i}$. We use linear and quadratic terms for binary model $f(y)$ instead of local fields and couplers, respectively. We also refer to inequalities \eqref{localFieldsRatio} and \eqref{couplerRatio} as the ratio inequalities for linear and quadratic terms, respectively. Notice that in these ratios, minimum and maximum coefficients are measured in magnitude, so in our discussion that follows, we consider merely the magnitude of the coefficients.

Considering the fact that $y_k^{x_i}y_l^{x_i}=y_l^{x_i}y_k^{x_i}$ and $y_k^{x_i}y_l^{x_j} = y_l^{x_j}y_k^{x_i}$, the coefficients of linear and quadratic terms are listed below:
\beqa
y_j^{x_i} & : & \l( Q_{ii} \l(c_j^{x_i}\r)^2 + q_i c_j^{x_i}  \r ),\\
y_k^{x_i}y_l^{x_i} & : & 2\l(Q_{ii} c_k^{x_i} c_l^{x_i}\r)  \quad  \text{for } \ k,l \in \{1, \ldots, d^{x_i}\} \  \text{ and } \ k<l,  \\
y_k^{x_i}y_l^{x_j} & :  & 2\l(Q_{ij} c_k^{x_i} c_l^{x_j}\r)  \quad \text{for }\  k \in \{1, \ldots, d^{x_i}\} ,\ l \in \{1, \ldots, d^{x_j}\}, \  \text{ and } \ i,j \in \{1, \ldots, n\} : i<j. \ \ \   
\eeqa
The difference between this case and what was presented in Section \ref{sec:findingUB} is that the coefficients of $y_j^{x_i}$ are no longer linear in $c_j^{x_i}$ (compare with \eqref{hint0}); therefore, the smallest coefficient may no longer occur at $c_j^{x_i}=1$. 
%%%%%%%%%%%
Notice that $Q_{ii} \l(c_j^{x_i}\r)^2 + q_i c_j^{x_i}$ intersects 0 at $c_j^{x_i}= 0$ and $c_j^{x_i} = \frac{-q_i}{Q_{ii}}$; therefore, if $\frac{-q_i}{Q_{ii}} < 1$, then $Q_{ii} \l(c_j^{x_i}\r)^2 + q_i c_j^{x_i}$ takes its minimum at $c_j^{x_i} =1$ and is increasing afterwards. 
When this is the case for all $i = 1, \ldots, n$, slight modification of  Algorithm \ref{alg:Find_$mu$_Ising} is sufficient to find the $\mu^{x_i}$'s. In the algorithm, the modification is at the initialization step of $\mu^{x_i}$'s. Letting
%%%
\beq
m_l = \min_{i} \left \{ |Q_{ii} + q_{i}| \right \} \ \text{ and } \ m_c = \min_{i,j} \left \{|Q_{ii}| , |Q_{ij}| \right \}, 
\eeq
the condition that needs to be satisfied for the linear coefficients is
\beqa
\frac{m_l} { |Q_{ii}| (\mu^{x_i})^2 + \sg(q_iQ_{ii}) |q_i| \mu^{x_i}} \ge \eps_l.
\eeqa
Combined with the condition
\beq
\mu^{x_i} \le  \sqrt{\frac{m_c}{|Q_{ii} |\eps_c} } , 
\eeq
we need to initialize $\mu^{x_i}$ as 
\beq
\mu^{x_i}= \l \lfloor \min \left \{ \tilde \mu^{x_i}, \sqrt{\frac{m_c}{|Q_{ii}| \eps_c} } \right \} \r \rfloor, 
\eeq
where
\beq
\tilde \mu^{x_i} = \frac{-\sg(q_iQ_{ii})|q_i| + \sqrt{|q_i|^2 +4 |Q_{ii}|\frac{m_l}{\eps_l}}}{2|Q_{ii}|}.
\eeq
It is worth mentioning the special cases where $\sg(Q_{ii}) = \sg(q_i)$ (resulting in $\frac{-q_i}{Q_{ii}} < 0$), and $Q_{ii}=0$ or $q_i= 0$ belong to the above category, where $Q_{ii} \l(c_j^{x_i}\r)^2 + q_i c_j^{x_i}$ attains its minimum at 1. 
%%%%%%%%%%%

In the general cases where there exists an $i$ such that $\frac{-q_i}{Q_{ii}}\ge 1$, satisfying the ratio condition on the linear terms is more complicated than what we discussed above. Not only might the minimum coefficient no longer occur at $1$, but also the maximum could occur at either $\mu^{x_i}$ or $\frac{-q_i}{2Q_{ii}}$, where the derivative of the function $h(c_j^{x_i}) = Q_{ii} \l(c_j^{x_i}\r)^2 + q_i c_j^{x_i}$ is zero. Our approach for these general cases is to first derive $\mu^{x_i}$'s that satisfy the ratio condition for the quadratic terms, and then adjust them accordingly to meet the ratio constraint on the linear terms.

Similar to the previous case, the minimum coefficient on the quadratic terms is
\beq\label{min_coup}
m_c = \min_{i,j} \left \{|Q_{ii}| , |Q_{ij}| \right \},
\eeq
and the ratio condition on quadratic terms enforces
\beq \label{min_coup_ineq}
\frac{m_c} { |Q_{ii}| (\mu^{x_i})^2 } \ge \eps_c, \ \text{and} \ \frac{m_c} { |Q_{ij}| \mu^{x_i}\mu^{x_j} } \ge \eps_c, 
\eeq
or, equivalently,
\beqa \label{cnd_ineq1}
\mu^{x_i} &\le & \sqrt{\frac{m_c}{|C_{ii}| \eps_c} } , \label{hintA0}\\ 
\mu^{x_i} \mu^{x_j} &\le & \frac{m_c}{|C_{ij}| \eps_c} \label{hintA1}.
\eeqa
We may now initialize $\mu^{x^i}$ as $\left \lfloor \sqrt{\frac{m_c}{|Q_{ii}| \eps_c}} \right \rfloor$ and use the loop of Algorithm \ref{alg:Find_$mu$_Ising} to satisfy \eqref{hintA1}, i.e., Algorithm \ref{alg:Find_$mu$_binary_part1}.
\begin{algorithm}
\caption{Finding $\mu^{x_i}$ for the Ratio Condition on Quadratic Terms}
\label{alg:Find_$mu$_binary_part1}
\begin{tabbing} 
\quad \quad \= \quad \quad  \=\quad \=\quad \=\kill \\
Inputs: \normalfont $\kappa, \ q,\ Q, \ \eps_c$\\
\> {\bf{set}} $m_c = \min_{i,j} \left \{ |Q_{ii}| , |Q_{ij}|\right \} $ \\
{Output:}\\
\> \normalfont  $\mu^{x_i}$ for $i= 1, 2, \ldots , n$ \\ 
\\
\normalfont {\bf{initialize}} $\mu^{x_i}= \l \lfloor \sqrt{\frac{m_c}{|Q_{ii}| \eps_c} } \r \rfloor$\\
\\
{\bf while } \normalfont any $\l(\mu^{x_i} \mu^{x_j} > \frac{m_c}{|Q_{ij}| \eps_c} \r)$\\
\\
\> \normalfont let $i,j = \arg \max_{i,j} \left \{ \mu^{x_i} \mu^{x_j} - \frac{m_c}{|Q_{ij}| \eps_c} \right \} $\\
\\
\> \normalfont let $\xi_i = \frac{\kappa^{x_i}}{ \mu^{x_i} -1} + \frac{\kappa^{x_j}}{ \mu^{x_j}}$ and  $\xi_j = \frac{\kappa^{x_i}}{ \mu^{x_i} } + \frac{\kappa^{x_j}}{ \mu^{x_j} -1}$\\
\\
\> \bf {if} $\xi_i < \xi_j$\\
\>\> $\mu^{x_i} =\mu^{x_i}-1$\\
\> \bf{else}\\
\>\> $\mu^{x_j} =\mu^{x_j}-1$	\\
\end{tabbing}
\end{algorithm}

After $\mu^{x_i}$'s are calculated to satisfy the conditions \eqref{hintA0} and \eqref{hintA1}, we need to check the ratio condition for the linear terms. Although some of the integer values $\{ 1, 2, \ldots, \mu^{x_i} \}$ may not appear in our encoding, knowing which integers will appear prior to finding $\mu^{x_i}$'s is not trivial. The algorithm presented below guarantees that the ratio condition on the linear terms holds if any of these integer values appear in the encoding. Let us categorize the indices based on where the minimum and maximum linear coefficients occur; we introduce the following sets of indices for this purpose:

\beqa
\mathcal I_0^m &=& \left \{ i: \sg(Q_{ii} ) \neq \sg( q_i ) \text{ and } \frac{-q_i}{Q_{ii}}=1 \right \}, \notag \\
\mathcal I_1^m &=& \left \{ i: \sg(Q_{ii} ) = \sg( q_i ) \right \} \cup \notag \\
&&  \left \{ i: \sg(Q_{ii} ) \neq \sg( q_i )  \text{ and } 0\le \frac{- q_i}{Q_{ii}} < 1 \right \} \cup \notag \\
&& \left \{i: \sg(Q_{ii} ) \neq \sg( q_i ) \text{ but } \mu^{x_i} < \left \lfloor \frac{- q_i}{Q_{ii}} \right \rfloor \right \}\cup \notag\\
&& \left \{ i: \sg(Q_{ii} ) \neq \sg( q_i ) \text{ but }  \frac{- q_i}{Q_{ii}} \text{ is an integer greater than 1}\right \}, \notag\\
\mathcal I_2^m &=& \left \{ i: \sg(Q_{ii} ) \neq \sg( q_i )  \text{ and }   \frac{- q_i}{Q_{ii}} > 1 \text{ and }  \frac{- q_i}{Q_{ii}} \text{ is not an integer }\text{ and } \mu^{x_i} \ge \left \lfloor \frac{- q_i}{Q_{ii}} \right \rfloor \right \}. \label{indList_m}
\eeqa
%{\color{red} { Need to add: when $\frac{- c_i}{C_{ii}} =1 $ }} DONE! 

and 
\beqa
\mathcal I_1^M &=& \left \{ i: \sg(Q_{ii} ) = \sg( q_i ) \right \} \cup \notag \\
&&  \left \{i: \sg(Q_{ii} ) \neq \sg( q_i )  \text{ and } 0\le \frac{- q_i}{2Q_{ii}} < \frac{1}{2} \right \} \cup \notag \\
&& \left \{i: \sg(Q_{ii} ) \neq \sg( q_i ) \text{ but } \mu^{x_i} \le \round \left(\frac{- q_i}{2Q_{ii}}\right) \right \}, \notag \\
\mathcal I_2^M &=& \left \{ i: \sg(Q_{ii} ) \neq \sg( q_i )  \text{ and }   \frac{- q_i}{2Q_{ii}} \ge \frac{1}{2} \text{ and } \mu^{x_i} > \round \left(\frac{- q_i}{2Q_{ii}}\right)\right \}. \label{indList_M}
\eeqa

Note that $\mu^{x_i}$'s returned by Algorithm \ref{alg:Find_$mu$_binary_part1} are integers. The sets with $m$ and $M$ superscripts are formed to facilitate computing minimum and maximum linear coefficients, respectively. For indices $i \in \mathcal I^m_0$, the minimum linear term happens at 2; for indices in $ \mathcal I^m_1$, it happens at 1, and for indices in $ \mathcal I^m_2$, it happens at one of the two integers closest to $\frac{-q_i}{Q_{ii}}$. Similarly, for the indices in $ \mathcal I^M_1$, the maximum linear coefficient happens at $\mu^{x_i}$, whereas for indices in $ \mathcal I^M_2$, it happens at either $\mu^{x_i}$ or at the closest integer to $\frac{-q_i}{2Q_{ii}}$, i.e., $\round \left(\frac{- q_i}{2Q_{ii}}\right)$. 

%Notice that $\mathcal I^m_2 \subset \mathcal I_2^M$; and $\mathcal I^M_1 \subset \mathcal I_1^m$. 

%Let $M^\mathcal I$ and $m^\mathcal I$ denote the maximum and minimum of set $\mathcal I$, respectively; and $M_i$ and $m_i$ denote the maximum and minimum of $i$th index. 
%removed : |Q_{ii} x^2 + q_i x|>0)
After categorizing the indices in the above sets, we may form the arrays $v^m$ and $v^M$, which represent the minimum and maximum coefficients of the linear terms for each variable, respectively: 
\beq
v^m_i = \begin{cases}
\left(|4Q_{ii} + 2q_i |, 2, i \right) & \text{if } i \in \mathcal I^m_0 \\
\left(|Q_{ii} + q_i |, 1, i \right) & \text{if } i \in \mathcal I^m_1 \\
\left( \min_{x \in \mathcal S_i } |Q_{ii} x^2 + q_i x|  ,  \argmin_{x \in\mathcal S_i } |Q_{ii} x^2 + q_i x|  , i \right) 
: \mathcal S_i = \l\{ x \in \left\{ \left \lfloor \frac{- q_i}{Q_{ii}} \right \rfloor, \left \lceil \frac{- q_i}{Q_{ii}} \right \rceil \right\} :  x\le \mu^{x_i} \r\} &  \text{if }  i \in \mathcal I^m_2 \label{vec_m}\\
\end{cases} 
\eeq

\beq
v^M_i = \begin{cases}
\left(|Q_{ii} (\mu^{x_i})^2 + q_i \mu^{x_i} |, \mu^{x_i}, i \right) & \text{if } i \in \mathcal I^M_1 \\
\left( \max_{x \in \mathcal S_i } |Q_{ii} x^2 + q_i x| , \argmax_{x \in \mathcal S_i} |Q_{ii} x^2 + q_i x|  , i\right) : \mathcal S_i =  \left\{ \round \left(\frac{- q_i}{2Q_{ii}}\right), \mu^{x_i} \right\}  & \text{if } i \in \mathcal I^M_2 \label{vec_M}
\end{cases} 
\eeq

We then sort entries of $v^m$ (on the first entry) in increasing order; assume it results in vector $\bar v^m$, so $\bar v_1^m \le \bar v_2^m \le \cdots \le \bar v_n^m$; and sort $v^M$ in decreasing order; i.e., $\bar v^M$ such that  $\bar v_1^M \ge \bar v_2^M \ge \cdots \ge \bar v_n^M$.  We then whether $\frac{\bar v_1^m}{\bar v_1^M} \ge \epsilon_l$. If this inequality holds, then we have obtained our set of $\mu^{x_i}$'s; otherwise, we have the following options for improvement:
\begin{itemize}
\item $\bar v^m_1$ happens at an $i \in \mathcal I_2^m$, where we can either change $\mu^{x_i}$ to increase the minimum coefficient, or change $\mu^{x_i}$ to decrease the maximum coefficient;

\item $\bar v^m_1$ happens at $i \in \mathcal I_1^m \cup \mathcal I_0^m$, in which case we can only change $\mu^{x_i}$ such that the maximum coefficient decreases.
\end{itemize}

In the first scenario where we have the option to both increase the minimum or decrease the maximum coefficient, we use a greedy approach to make the decision. In other words, if the minimum coefficient is improved, it will be $v^m_{2,1}$ in the next iterate; similarly, the maximum coefficient will $\bar v^M_{2,1}$, if updated. In our approach, having the interval $[\bar v^m_{1,1} , \bar v^M_{1,1}]$ for the coefficients, we wish to update it to either $[\bar v^m_{2,1} , \bar v^M_{1,1}]$ or $[\bar v^m_{1,1} , \bar v^M_{2,1}]$. These two intervals will be $[\frac{\bar v^m_{2,1}}{\bar v^M_{1,1}} , 1]$ or $[\frac{\bar v^m_{1,1}}{ \bar v^M_{2,1}}, 1]$, respectively, after the rescaling. We choose the option that gives us better lower bound, i.e., if 
$$\frac{\bar v^m_{2,1}}{\bar v^M_{1,1}} > \frac{v^m_{1,1}}{ \bar v^M_{2,1}}  \  \equiv \  \bar v^m_{2,1} \bar v^M_{2,1} > \bar v^m_{1,1} \bar v^M_{1,1}\,, $$
we attempt to improve the lower bound, thus decreasing $\mu^{x_i}$ for $i = \bar v^m_{1,3}$.
%that can be changed to $[\bar a, b]$ ($\equiv [\frac{\bar a}{b} ,1]$) or $[a, \bar b]$ ($\equiv [\frac{ a}{\bar b} ,1]$), we choose the option with larger lower end; i.e. if $\frac{\bar a}{b} > \frac{ a}{ \bar b} (\equiv  \bar b \bar a > ab)$ we update the lower bound else we update the upper bound. 
A formal presentation of what we have discussed is summarized in Algorithm \ref{Confirm $mu$'s}.

\begin{algorithm}
\caption{Adjusting $\mu^{x_i}$ for Ratio Condition on Linear Terms}
\label{Confirm $mu$'s}
\begin{tabbing} 
\quad \quad \= \quad \quad  \=\quad \quad \=\quad \quad \=\quad \quad \=\kill \\
Inputs:\\
\> \normalfont $ q_{i}, \ Q_{ii}, \ \eps_l$, and $\mu^{x_i} $ that is the output of Algorithm \ref{alg:Find_$mu$_binary_part1}
\\
{Output:}\\
\> \normalfont  $\mu^{x_i}$ for $i= 1, 2, \ldots , n$ \\ 
\\
\normalfont {\bf{initialize}} $\mathcal I^m_0, \mathcal I^m_1, \mathcal I^m_2, \mathcal I^M_1, \text{ and } \mathcal I^M_2$, using equations \eqref{indList_m} and \eqref{indList_M}.\\
\\
{\bf while 1}\\
\> \normalfont form $v^m$ and $v^M$ using equation \eqref{vec_m} and \eqref{vec_M}. \\
\> \normalfont sort $v^m$ increasingly,  and $v^M$ decreasingly (on the first entry) to get $\bar v^m$ and $\bar v^M$\\
\\
\> \normalfont {\bf if} $\frac{\bar v^m_{11}}{\bar v^M_{11}} \ge \epsilon_\ell$ {\bf return}\\
\> \normalfont {\bf else}\\
\>\> \normalfont { \bf if} $k = \bar v^m_{1,3} \in \mathcal I^m_2$  AND  $\bar v^m_{2,1} \bar v^M_{2,1} > \bar v^m_{1,1} \bar v^M_{1,1}$  \\
\\
\>\>\> \normalfont {{set $\mu^{x_k} = \mu^{x_k} -1$ ( can also be $\bar v_{12}^m-1$ )   ,}} \\
\\
\>\>\> \normalfont {\bf if} $\mu^{x_k} < \l \lfloor \frac{- q_k}{Q_{kk}} \r \rfloor$ ( OR $\bar v_{1,2}^m =\l \lfloor \frac{- q_k}{Q_{kk}} \r \rfloor$)\\
\\
\>\>\>\> $\mathcal I_2^m =\mathcal I_2^m \setminus \{ k\} $ and $\mathcal I_1^m =\mathcal I_1^m \cup \{ k\} $   \\
\>\> \normalfont {\bf else} \\
\>\>\> $k=\bar v^M_{1,3}$\\
\\
\>\>\> \normalfont {\bf if} $k\in \mathcal I^M_2$ AND $\bar v^M_{1,2} = \round\left(\frac{-q_k}{2Q_{kk}}\right)$\\
\>\>\>\> \normalfont set $\mu^{x_k} = \round\left(\frac{-q_k}{2Q_{kk}}\right) -1$, $\mathcal I_2^M =\mathcal I_2^M \setminus \{ k\} $ and $\mathcal I_1^M =\mathcal I_1^M \cup \{ k\} $   \\
\>\>\> \normalfont {\bf else} \\
\>\>\> \> set $\mu^{x_k} = \mu^{x_k} -1$
\\
\end{tabbing}
\end{algorithm}

%%%%%%%%%%%
%%%%%%%%%%%

\section*{Acknowledgements}
We are thankful to Helmut G. Katzgraber and Gili Rosenberg for insightful discussions and helpful feedback; and Marko Bucyk for editing the manuscript.

%%%%%%%%%%%
%%%%%%%%%%%
%%%%%%%%%%%
\newpage
%\bibliographystyle{plain}
%\bibliography{references}

\begin{thebibliography}{10}

%\bibitem{farhi00}
%E.~{Farhi}, J.~{Goldstone}, and S.~{Gutmann}.
%\newblock {A Numerical Study of the Performance of a Quantum Adiabatic
%  Evolution Algorithm for Satisfiability}.
%Jul. 2000.
%\newblock \href {}
%  {\path{arXiv:}}.

\bibitem{albash2016}
T.~Albash and D.~A. Lidar.
\newblock {Adiabatic quantum computing}.
Nov. 2016
\newblock \href {http://arxiv.org/abs/1611.04471}
{\path{arXiv:quant-ph/1611.04471}}.

\bibitem{albash2014}
T.~Albash, W.~Vinci, A.~Mishra, P.~A. Warburton, and D.~A. Lidar.
\newblock Consistency tests of classical and quantum models for a quantum
  annealer.
\newblock {\em Phys. Rev. A}, 91:042314, Apr. 2015.
\newblock  \href  {http://dx.doi.org/10.1103/PhysRevA.91.042314}
  {\path{doi:10.1103/PhysRevA.91.042314}}.

\bibitem{AQC_qudits}
Mohammad Amin, Neil Dickson, and Peter Smith.
\newblock Adiabatic quantum optimization with qudits.
\newblock {\em Quantum Information Processing}, 12(4):1819 -- 1829, 2013.
\newblock  \href  {http://dx.doi.org/10.1007/s11128-012-0480-x}
  {\path{doi:10.1007/s11128-012-0480-x}}.

\bibitem{chimera}
P.~I. Bunyk, E.~M. Hoskinson, M.~W. Johnson, E.~Tolkacheva, F.~Altomare, A.~J.
  Berkley, R.~Harris, J.~P. Hilton, T.~Lanting, A.~J. Przybysz, and
  J.~Whittaker.
\newblock Architectural considerations in the design of a superconducting
  quantum annealing processor.
\newblock {\em IEEE Transactions on Applied Superconductivity}, 24(4):1--10,
  Aug. 2014.
\newblock \href {http://dx.doi.org/10.1109/TASC.2014.2318294}
  {\path{doi: 10.1109/TASC.2014.2318294}}.

\bibitem{AidanRoy2014}
J.~{Cai}, W.~G. {Macready}, and A.~{Roy}.
\newblock {A practical heuristic for finding graph minors}.
Jun. 2014.
\newblock \href {http://arxiv.org/abs/1406.2741} {\path{arXiv:1406.2741}}.

\bibitem{farhi01}
E.~Farhi, J.~Goldstone, S.~Gutmann, J.~Lapan, A.~Lundgren, and D.~Preda.
\newblock A quantum adiabatic evolution algorithm applied to random instances
  of an {NP}-complete problem.
\newblock {\em Science}, 292(5516):472--476, 2001.
\newblock 
\href{http://dx.doi.org/10.1126/science.1057726}
  {\path{doi:10.1126/science.1057726}}.

\bibitem{Ishikawa2011}
H.~Ishikawa.
\newblock Transformation of general binary {MRF} minimization to the first-order
  case.
\newblock {\em IEEE Transactions on Pattern Analysis and Machine Intelligence},
  33(6):1234--1249, Jun. 2011.
\newblock \href {http://dx.doi.org/10.1109/TPAMI.2010.91}
  {\path{doi:10.1109/TPAMI.2010.91}}.

\bibitem{DwaveNature}
M.~W. Johnson, M.~H.~S. Amin, S.~Gildert, T.~Lanting, F.~Hamze, N.~Dickson,
  R.~Harris, A.~J. Berkley, J.~Johansson, P.~Bunyk, E.~M. Chapple, C.~Enderud,
  J.~P. Hilton, K.~Karimi, E.~Ladizinsky, N.~Ladizinsky, T.~Oh, I.~Perminov,
  C.~Rich, M.~C. Thom, E.~Tolkacheva, C.~J.~S. Truncik, S.~Uchaikin, J.~Wang,
  B.~Wilson, and G.~Rose.
\newblock Quantum annealing with manufactured spins.
\newblock {\em Nature}, 473(7346):194--198, 05 2011.
\newblock \href{http://dx.doi.org/10.1038/nature10012}
{\path{doi:10.1038/nature10012}}.

\bibitem{karimi_GD}
S.~Karimi and P.~Ronagh.
\newblock A subgradient approach for constrained binary programming via quantum
  adiabatic evolution.
Jan. 2017.
\newblock \href {https://arxiv.org/abs/1605.09462}
{\path{arXiv:1605.09462}}.

\bibitem{katzgraber2015}
H.~G. Katzgraber, F.~Hamze, Z.~Zhu, A.~J. Ochoa, and H.~Munoz-Bauza.
\newblock Seeking quantum speedup through spin glasses: The good, the bad, and
  the ugly.
\newblock {\em Phys. Rev. X}, 5:031026, Sep. 2015.
\newblock  \href  {http://dx.doi.org/10.1103/PhysRevX.5.031026}
  {\path{doi:10.1103/PhysRevX.5.031026}}.
  
\bibitem{katzgraber_qudits}
L.~W. Lee, H.~G. Katzgraber, and A.~P. Young.
\newblock Critical behavior of the three- and ten-state short-range potts
  glass: A monte carlo study.
\newblock {\em Phys. Rev. B}, 74:104416, Sep 2006.
\newblock  \href  {http://dx.doi.org/10.1103/PhysRevB.74.104416}
  {\path{doi:10.1103/PhysRevB.74.104416}}.

\bibitem{biasedGSsampling}
Salvatore Mandr\`a, Zheng Zhu, and Helmut~G. Katzgraber.
\newblock Exponentially biased ground-state sampling of quantum annealing
  machines with transverse-field driving Hamiltonians.
\newblock {\em Phys. Rev. Lett.}, 118:070502, Feb 2017.
\newblock  \href  {https://doi.org/10.1103/PhysRevLett.118.070502}
  {\path{doi:10.1103/PhysRevLett.118.070502}}.

\bibitem{portfolioOpt}
R.~Mansini, W.~Ogryczak, and M.~G. Speranza.
\newblock {\em Linear Models for Portfolio Optimization}.
\newblock Springer International Publishing: 19--45, 2015.
\newblock  \href  {https://doi.org/10.1007/978-3-319-18482-1_2}
  {\path{doi:10.1007/978-3-319-18482-1_2}}.

\bibitem{GSstat}
Yoshiki Matsuda, Hidetoshi Nishimori, and Helmut~G Katzgraber.
\newblock Ground-state statistics from annealing algorithms: quantum versus
  classical approaches.
\newblock {\em New Journal of Physics}, 11(7):073021, 2009.
\newblock  \href  {https://doi.org/10.1088/1367-2630/11/7/073021}
  {\path{doi:10.1088/1367-2630/11/7/073021}}.

\bibitem{McGeoch2013}
C.~C. McGeoch and C.~Wang.
\newblock Experimental evaluation of an adiabiatic quantum system for
  combinatorial optimization.
\newblock In {\em Proceedings of the ACM International Conference on Computing
  Frontiers}, CF '13, pages 23:1--23:11, New York, NY, USA, 2013. ACM.
\newblock 
  \href {http://dx.doi.org/10.1145/2482767.2482797}
  {\path{doi:10.1145/2482767.2482797}}.
  
\bibitem{Venturelli2015}
E.~G. {Rieffel}, D.~{Venturelli}, B.~{O'Gorman}, M.~B. {Do}, E.~M. {Prystay},
  and V.~N. {Smelyanskiy}.
\newblock {A case study in programming a quantum annealer for hard operational
  planning problems}.
\newblock {\em Quantum Information Processing}, 14:1--36, Jan. 2015.
\newblock   \href {http://dx.doi.org/10.1007/s11128-014-0892-x}
  {\path{doi:10.1007/s11128-014-0892-x}}.

\bibitem{Gili}
G.~Rosenberg, P.~Haghnegahdar, P.~Goddard, P.~Carr, K.~Wu, and M.~L. de~Prado.
\newblock Solving the optimal trading trajectory problem using a quantum
  annealer.
\newblock {\em IEEE Journal of Selected Topics in Signal Processing},
  10(6):1053--1060, Sep. 2016.
\newblock \href {http://dx.doi.org/10.1109/JSTSP.2016.2574703}
  {\path{doi:10.1109/JSTSP.2016.2574703}}.

\bibitem{optimization_AQO}
Giuseppe~E Santoro and Erio Tosatti.
\newblock Optimization using quantum mechanics: quantum annealing through
  adiabatic evolution.
\newblock {\em Journal of Physics A: Mathematical and General}, 39(36):R393,
  2006.
\newblock \href {http://dx.doi.org/10.1088/0305-4470/39/36/R01}
{\path{doi:10.1088/0305-4470/39/36/R01}}.

\bibitem{supplyChain}
Tadeusz Sawik.
\newblock {\em Scheduling in Supply Chains Using Mixed Integer Programming}.
\newblock John Wiley \& Sons Inc., 2011.
\newblock \href {http://dx.doi.org/10.1002/9781118029114}
{\path{doi:10.1002/9781118029114}}.

\bibitem{Dominic2015}
D.~{Venturelli}, D.~J.~J. {Marchand}, and G.~{Rojo}.
\newblock {Quantum annealing implementation of job-shop scheduling}.
Jun. 2015.
\newblock \href {http://arxiv.org/abs/1506.08479} {\path{arXiv:1506.08479}}.

\bibitem{zoshk2016}
Z.~Zhu, A.~J. Ochoa, S.~Schnabel, F.~Hamze, and H.~G. Katzgraber.
\newblock Best-case performance of quantum annealers on native spin-glass
  benchmarks: How chaos can affect success probabilities.
\newblock {\em Phys. Rev. A}, 93:012317, Jan. 2016.
\newblock \href  {http://dx.doi.org/10.1103/PhysRevA.93.012317}
  {\path{doi:10.1103/PhysRevA.93.012317}}.

\bibitem{Zick}
K.~M. Zick, O.~Shehab, and M.~French.
\newblock Experimental quantum annealing: Case study involving the graph
  isomorphism problem.
\newblock {\em Scientific Reports}, 5:11168, Jun. 2015.
\newblock \href {http://dx.doi.org/10.1038/srep11168}
  {\path{doi:10.1038/srep11168}}.

\end{thebibliography}

\end{document}